\documentclass[dvips,a4paper,10pt]{article}
\usepackage{epsfig,graphicx,verbatim,amsmath,amsfonts,amssymb,epsf,epic,alltt,algorithm,algorithmic,relsize}
\newtheorem{theorem}{Theorem}[section]
\newtheorem{proposition}[theorem]{Proposition}

\newtheorem{lemma}[theorem]{Lemma}

\newtheorem{corollary}[theorem]{Corollary}

\def\boldhead#1:{\par\vskip 7pt\noindent{\bf #1:}\hskip 10pt}
\def\ithead#1:{\par\vskip 7pt\noindent{\it #1:}\hskip 10pt}

\def\inline#1:{\par\vskip 7pt\noindent{\bf #1:}\hskip 10pt}
\def\midinline#1:{\par\noindent{\bf #1:}\hskip 10pt}
\def\dnsinline#1:{\par\vskip -7pt\noindent{\bf #1:}\hskip 10pt}
\def\ddnsinline#1:{\newline{\bf #1:}\hskip 10pt}
\def\largeinline#1:{\par\vskip 7pt\noindent{\large\bf #1:}\hskip 10pt}
\long\def\comment #1\commentend{}
\long\def\commhide #1\commhideend{}
\long\def\commfull #1\commend{#1}
\long\def\commabs #1\commenda{}
\long\def\commtim #1\commendt{#1}
\long\def\commb #1\commbend{}
\long\def\commedit #1\commeditend{} 

\long\def\commB #1\commBend{}       

\long\def\commex #1\commexend{}     

\long\def\commsiena #1\commsienaend{}  
                                         
\long\def\commBI #1\commBIend{}  
                                         
\long\def\CProof #1\CQED{}

\def\blackslug{\hbox{\hskip 1pt \vrule width 4pt height 8pt
    depth 1.5pt \hskip 1pt}}
\def\QED{\quad\blackslug\lower 8.5pt\null\par}

\def\Proof{\par\noindent{\bf Proof:~}}

\def\proof{\Proof}

\long\def\PPP#1{\noindent{\bf Proof:}{ #1}{\quad\blackslug\lower 8.5pt\null}}
\long\def\PP#1{\noindent {\bf Proof}:{ #1} $\Box$ \\}

\long\def\denspar #1\densend
{#1}

\newif\ifnotesw\noteswtrue
   {\ifnotesw\marginpar[\hfill\(\top\)]{\(\top\)}\fi}%
   {\ifnotesw\marginpar[\hfill\(\bot\)]{\(\bot\)}\fi}

\newcommand{\mnote}[1]%
    {\ifnotesw\marginpar%
        [{\scriptsize\it\begin{minipage}[t]{\marginparwidth}
        \raggedleft#1%
                        \end{minipage}}]%
        {\scriptsize\it\begin{minipage}[t]{\marginparwidth}
        \raggedright#1%
                        \end{minipage}}%
    \fi}

\def\MathF{\hbox{\rm I\kern-2pt F}}
\def\MathP{\hbox{\rm I\kern-2pt P}}
\def\MathR{\hbox{\rm I\kern-2pt R}}
\def\MathZ{\hbox{\sf Z\kern-4pt Z}}
\def\MathN{\hbox{\rm I\kern-2pt I\kern-3.1pt N}}
\def\MathC{\hbox{\rm \kern0.7pt\raise0.8pt\hbox{\footnotesize I}
\kern-4.2pt C}}
\def\MathQ{\hbox{\rm I\kern-6pt Q}}

\newsavebox{\ttop}\newsavebox{\bbot}

\begin{document}
\newtheorem{question}{Question}
\newtheorem{lem}{Lemma}

\newcommand {\ignore} [1] {}

\def \bfm {\bf \boldmath}

\def \AA {{\alpha}}
\def \BB {{\beta}}
\def \CC {{\gamma}}
\def \PP {{\cal P}}

\def \pa {{\boldmath $A$}}
\def \pb {{\boldmath $B$}}
\def \pc {{\boldmath $C$}}
\def \px {{\boldmath $X$}}
\def \py {{\boldmath $Y$}}
\def \pz {{\boldmath $Z$}}

\fontfamily{cmss} \selectfont
\baselineskip 20pt \centerline{\bf \Large The Tower of Hanoi problem on \texttt{Path$_{h}$} graphs} \vskip
15pt
{\sc Daniel Berend} \\
\small {\it berend@cs.bgu.ac.il} \\
{\it Departments of Mathematics and Computer Science, Ben-Gurion University, Beer-Sheva, Israel} \\ \normalsize
\vskip 12pt
{\sc Amir Sapir} \\
\small {\it amirsa@cs.bgu.ac.il} \\
{\it Department of Software Systems, Sapir College, Western Negev, Israel \footnote{\normalsize \fontfamily{cmss} \selectfont Research supported in part by the Sapir Academic College, Israel.} and} \\
{\it The Center for Advanced Studies in Mathematics at Ben-Gurion University, Beer-Sheva, Israel
} \\ \normalsize
\vskip 12pt
{\sc Shay Solomon} \\
\small {\it shayso@cs.bgu.ac.il} \\
{\it Department of Computer Science, Ben-Gurion University, Beer-Sheva, Israel} \\ \normalsize

\fontfamily{cmss} \selectfont


\begin {abstract}
The generalized Tower of Hanoi problem with $h \ge 4$ pegs is known to require a sub-exponentially fast growing number of moves in order to
transfer a pile of $n$ disks from one peg to another. In this paper
we study the \texttt{Path$_{h}$} variant, where the pegs are placed
along a line, and disks can be moved from a peg to its nearest
neighbor(s) only.

Whereas in the simple variant there are $h(h-1)/2$
possible bi-directional interconnections among pegs, here there are
only $h-1$ of them. Despite the significant reduction in the number
of interconnections, the number of moves needed to transfer a pile
of $n$ disks between any two pegs also grows
sub-exponentially as a function of $n$.

We study these graphs, identify sets of mutually recursive tasks, and
obtain a relatively tight upper bound for the number of moves, depending on $h, n$ and the source and destination pegs.
\end {abstract}

Keywords: {\it Tower of Hanoi, path graphs, analysis of algorithms}
\pagenumbering {arabic}
\large \baselineskip 18pt
\section{Introduction}~\label{sec1}
In the well-known Tower of Hanoi problem, proposed over a hundred
years ago by Lucas~\cite{Lucas1893}, a player is given 3 pegs and a
certain number $n$ of disks of distinct sizes, and is required to
transfer them from one peg to another. Initially all disks are
stacked (composing a tower) on the first peg (the source) ordered
monotonically by
size, with the smallest at the top and the largest at the bottom.
The goal is to transfer them to the third peg (the destination),
moving only topmost disks, and never placing a disk on top of a
smaller one. The well-known recursive algorithm that accomplishes
this task requires $2^{n}-1$ steps, and is the unique optimal
algorithm for the problem. The educational aspects of the Tower of Hanoi puzzle have been reinforced recently, by a series of papers by Minsker (\cite{Mins2007,Mins2008,Mins2009}), composing variants for the sake of studying their combinatorial as well as algorithmic aspects.

Work on this problem still goes on, studying properties of solution
instances, as well as variants of the original problem. Connections
between Pascal's triangle, the Sierpi\'{n}ski gasket and the Tower
of Hanoi are established in \cite{Hinz92}, and to some classical numbers in \cite{KlMiPe2005}. In \cite{AlAsRaSha94} it is shown that, with a certain way of coding the moves, a string
which represents an optimal solution is square-free. This line is extended in \cite{AlSa2005}. Another
direction was concerned with various generalizations, such as having
any initial and final configurations \cite{Er85}, assigning colors
to disks (cf. \cite{LuSt2008} and \cite{Mins2005} for recent papers on the subject), and relaxing the placement rule of disks by allowing a disk to be placed on top of a smaller one under
prescribed conditions \cite{DiSo2006,DiSo2007b,DiSo2007}.

A natural extension of the original problem is obtained by adding pegs. One of the earliest versions is ``The Reve's Puzzle" \cite[pp. 1-2]{Dude08}. There it was presented in a limited form: $4$ pegs and specified numbers of disks. The general setup of the problem, with any number $h > 3$ of pegs and any number of disks, was suggested in
\cite{Stew39}, with solutions in \cite{Stew41} and \cite{Frame41},
shown recently to be identical \cite{KlMiPe2002}. An analysis of the
algorithm reveals, somewhat surprisingly, that the solution grows
sub-exponentially, at the rate of $\Theta(\sqrt{n}2^{\sqrt{2n}})$
for $h=4$ (cf. \cite{Stoc94}). The lower bound issue was considered
in \cite{Szeg99} and \cite{ChSh2004}, where it has been shown that the minimal number of moves
grows roughly at the same rate.

An imposition of movement restrictions among pegs generates many
variants, and calls for representing variants by digraphs, where a
vertex designates a peg, and an arc represents the permission to
move a disk in the appropriate direction. In \cite{Atki81,Er84}, the uni-directional cyclic 3-peg variant (\texttt{Cyclic$_{3}$}) has been studied, and the average distance between the nodes -- in \cite{Stoc96}. In \cite{ScGrSm44}, the
``three-in-a-row" arrangement (\texttt{Path$_{3}$}) is
discussed. A unified treatment of all 3-peg variants is given in \cite{Sapir2004}. The (uni-directional) \texttt{Cyclic$_{4}$} is discussed for the first time in \cite{ScGrSm44}, and \cite{Stoc94} studies other $4$-peg variants: \texttt{Star$_{4}$} and \texttt{Path$_{4}$}, presenting a sub-exponential algorithm for
\texttt{Star$_{4}$}.
The \texttt{Cyclic}$_{h}$ for any number of pegs $h \ge 4$ has been studied in \cite{besa2005a} and proved to be exponential for any specified $h$. Identification of the longest task, for certain variants, has been resolved in \cite{BeSa2005b}.

The only requirement for the problem to be solved for any number of disks is that the variant is represented by a strongly-connected directed graph. An interesting line of research has been taken in \cite{Leiss84}, \cite{AzBe2006}, and \cite{AzSoSo2008}, where non-strongly-connected graphs are being studied.

In this paper we study the \texttt{Path$_h$} variant, with a fixed number $h \ge
4$ of pegs, whose complexity issue has been left open. We devise an efficient algorithm which moves a column of $n$ disks between any pair of pegs, and supply an explicit subexponential upper bound on the number of moves, for each $h$.

Notations and definitions are given in Section~\ref{sec2}, the main results in Section~\ref{sec3}, the proof of the 4-peg case in Section~\ref{sec4} and that of the general case in Section~\ref{sec5}.

\section{Preliminaries}~\label{sec2}
We study the \texttt{Path$_h$} (a.k.a. \emph{$h$-in-a-row})
variant, with a fixed number $h \ge
4$ of pegs. We denote the pegs of
\texttt{Path$_h$}, from left to right, by $1,\ldots,{h}$. Let the sizes of the disks be $1, 2, \ldots, n$. For
convenience, we identify the name of a disk with its size.

For the statements and algorithms of the paper, it is required to introduce the notion of a {\it block} --- a set of disks of consecutive sizes. The minimum (respectively, maximum)
size of a disk in a block $B$ is denoted by $B_{\min}$ (resp.,
$B_{\max}$), and the number $B_{\max} - B_{\min} + 1$ of disks in $B$ --- by $|B|$. A block $B$ is \emph{lighter} than
another block $B'$ if $B_{\max} < B'_{\min}$.

A {\it configuration} is a legal distribution of the $n$ disks
among the $h$ pegs. A~{\it perfect configuration} is one in which
all the disks reside on the same peg. Such a configuration is
denoted by $R_{h,i,n}$, where $h$ is the
number of pegs, $i$ the peg holding the disks, and $n$ the number of disks.

For a sequence of moves $M$, henceforth \emph{move-sequence}, we
denote by $M^{-1}$ the reverse move-sequence, comprising the
moves that cause the reverse effect. That is, the order of the moves is reversed and each move of the original sequence is reversed.
Clearly, if applying $M$ to configuration $C_{1}$ results in reaching
configuration $C_{2}$, then applying $M^{-1}$ to $C_{2}$ results in configuration $C_{1}$.
(Note that this is true if and only if the peg structure is a graph; for digraphs in general this is not true.)

A problem instance, henceforth a {\it task}, is given by a pair of configurations, an
{\it initial} configuration $C_1$ and a {\it final} configuration
$C_2$, where we are required to move from $C_1$ to $C_2$ in a minimal
number of moves. The task, as well as a minimal-length solution of it, is denoted by $C_1
\rightarrow C_2$, and the minimum number of moves needed to get from
$C_1$ to $C_2$ is denoted by $|C_1 \rightarrow C_2|$.

In this paper we focus on {\it perfect tasks} --- problem instances whose
initial and final configurations are both perfect. The peg associated
with the initial (respectively, final) configuration of a
perfect task is naturally referred to as the
{\it source} (resp., {\it destination}). Clearly, for any positive
integers $h, n$ and $1 \le i < j \le h$, we have $|R_{h,i,n}
\rightarrow R_{h,j,n}| = |R_{h,j,n} \rightarrow R_{h,i,n}|$. We
shall henceforth restrict our attention in the sequel to
tasks in which the source peg is
to the left (i.e. has a lower peg index) of the destination peg.

For $n \ge 1$, denote by Path($h,n$) the minimal
number of moves which suffices for transferring a block of size $n$ between all pairs
of perfect configurations in \texttt{Path$_h$}, namely,
$${\rm Path}(h,n) = \max_{1 \le i < j \le h} |R_{h,i,n} \rightarrow
R_{h,j,n}| \, .$$

For a real number $x$, let round($x$) be the
 integer closest to $x$ (where round($x$) = $\lceil x \rceil$ for $x = n + 0.5$).
For a pair of positive integers $p$ and $q$, with $p < q$, we denote the
set $\{p,\ldots,q\}$ by $[p,q]$, and $[1, \ldots, q]$ by $[q]$.
In what follows, we do not distinguish between a move-sequence and
an algorithm generating it, if this does not lead to a
misunderstanding.

\section{Main results}~\label{sec3}
The main question the paper addresses is: what is the complexity of {\rm Path($h,n$)}?
\\An upper bound is provided by

\begin{theorem}~\label{tp1} {\rm Path($h,n$)} $\le C_{h}n^{\alpha_{h}}\cdot 3^{\theta_{h}\cdot n^{\frac{1}{h-2}}} ,$  for all $h \ge 3$ and $n$, where:

\hspace{6mm}$\theta_{h} = ((h-2)!)^{\frac{1}{h-2}}$ ,

\hspace{6mm}$\alpha_{h} = \frac{h-3}{h-2}$ ,

\hspace{6mm}$C_{h} = \frac{(h-2) \cdot \delta^{h-3}}{\theta_{h}}$, $\qquad \left(\delta = \frac{11}{3^{2-(1/30)^{1/3}}}\right)$.

In particular, {\rm Path($h,n$)} grows subexponentially as a function of $n$ for $h \ge 4$.
\end{theorem}

Of course, as a lower bound for {\rm Path($h,n$)} one may use any lower bound for the number of moves required to move a tower of size $n$ from one peg to another over the complete graph on $h$ vertices. By \cite{ChSh2004}, such a lower bound is given by $2^{(1+o(1))(n(h-2)!)^{\frac{1}{h-2}}}$, which is ``not very far" from our upper bound for {\rm Path($h,n$)}.

The following theorem identifies the hardest perfect task for the particular case $h=4$. It also provides a tighter upper bound for
\texttt{Path$_4$} than the one given in Theorem~\ref{tp1}.
\begin{theorem}~\label{tp2}
For every $n \ge 1$:
\begin{description}
\item {\rm(a)} $|R_{4,i,n} \rightarrow R_{4,j,n}| ~<~ |R_{4,1,n} \rightarrow R_{4,4,n}| \,\, for \,\, 1 \le i < j \le 4, (i,j) \neq (1,4)$.
In particular, {\rm Path($4,n$)} = $|R_{4,1,n} \rightarrow R_{4,4,n}|$.
\item {\rm(b)} {\rm Path($4,n$)} $< 1.6 \sqrt{n} 3^{\sqrt{2n}}$.
\end{description}
\end{theorem}

\section{Proof of Theorem \ref{tp2}} \label{sec4}
\subsection{On the relation between various tasks in \texttt{Path$_4$}}
We start with a result of some independent interest, which holds for general $h$.
\begin{lemma}\label{rel4path01}
Let $C$ be a configuration with $n \ge 1$ disks, arranged
arbitrarily on pegs $1,\ldots,h-2$, with pegs $h-1$ and $h$ empty. Then:
$$|C \rightarrow R_{h,h-2,n}| < |C \rightarrow R_{h,h,n}|, \qquad |C \rightarrow R_{h,h-1,n}| < |C \rightarrow R_{h,h,n}| \, .$$
\end{lemma}


\begin {proof}
We detail the proof for $|C \rightarrow
R_{h,h-2,n}| < |C \rightarrow R_{h,h,n}|$. The proof for the second inequality is similar.

The proof is by induction on $n$. The basis $n=1$ is trivial. Let $n \ge 2$, assume that the statement holds for up to $n - 1$ disks, and let $C$ be a configuration as in the statement of the lemma and $M$ a move-sequence transferring from $C$ to $R_{h,h,n}$. Before the last move
of disk $n$ (to peg $h$), a configuration $C'$, in which all $n-1$
disks $1, 2, \ldots, n-1$ are distributed among pegs $1,\ldots,h-2$, is reached.
Let $M'$ (respectively, $M''$) be the subsequence of $M$, consisting of
all moves that come before (resp., after) the last move of disk $n$.
Notice that $M''$ transfers from $C'$ (considered as a configuration of $n-1$ disks)  to $R_{h,h,n-1}$.
By the induction hypothesis, there exists a move-sequence $M''_{h-2}$ that transfers
from $C'$ to $R_{h,h-2,n-1}$, which
is strictly shorter than $M''$. Let $M'_{h-2}$ be the move-sequence obtained from $M'$ by omitting all moves
disk $n$ makes after reaching peg $h-2$ for the first
time. Concatenating $M'_{h-2}$ with $M''_{h-2}$, we obtain a legal move-sequence, strictly shorter than
$M$, transferring from $C$ to $R_{h,h-2,n}$. The required
result follows. 
\end {proof}

Due to symmetries, there are actually only four essentially distinct perfect tasks in \texttt{Path$_4$}: $R_{4,1,n} \rightarrow
R_{4,2,n}$, $R_{4,1,n} \rightarrow R_{4,3,n}$, $R_{4,1,n}
\rightarrow R_{4,4,n}$, and $R_{4,2,n} \rightarrow R_{4,3,n}$. By
Lemma~\ref{rel4path01}, taking $h=4$ and $C$ to be various perfect configurations, we obtain for any $n \ge 1$
\begin {itemize}
\item $|R_{4,1,n} \rightarrow R_{4,2,n}| ~<~ |R_{4,1,n} \rightarrow R_{4,4,n}|$.
\item $|R_{4,2,n} \rightarrow R_{4,3,n}|  ~<~ |R_{4,2,n} \rightarrow R_{4,4,n}| ~=~ |R_{4,1,n} \rightarrow R_{4,3,n}|
~<~ |R_{4,1,n} \rightarrow R_{4,4,n}|$.
\end {itemize}

These inequalities establish part (a) of Theorem~\ref{tp2}.

In Table~1 we present the (distinct) numbers $|R_{4,i,n} \rightarrow R_{4,j,n}|$ for $1 \le n \le 11$. The entries have been calculated by finding the distance between the vertices $R_{4,i,n}$ and $R_{4,j,n}$ in the graph of all configurations of $n$ disks on \texttt{Path$_{4}$} using breadth-first search.

\begin{table}[ht]~\label{tab1}
$$\begin{tabular}{|c|r|r|r|r|r|}
\hline
       & \multicolumn{5}{c|}{tasks} \cr
\hline
 disks \vrule height 16pt depth 6pt width 0pt & $2 \rightarrow 3$ & $1 \rightarrow 2$ & $1 \rightarrow 3$ & $1 \rightarrow 4$ & $\frac{1 \rightarrow 4}{\sqrt{n}3^{\sqrt{2n}}}$ \vrule depth 10pt width 0pt \cr
\hline
      &     &     &     &     & \cr
  1   &   1 &   1 &   2 &   3 & 0.634                                  \cr
  2   &   4 &   4 &   6 &  10 & 0.786                                  \cr
  3   &   7 &   9 &  12 &  19 & 0.744                                  \cr
  4   &  14 &  18 &  22 &  34 & 0.760                                  \cr
  5   &  23 &  29 &  36 &  57 & 0.790                                  \cr
  6   &  34 &  44 &  54 &  88 & 0.799                                  \cr
  7   &  53 &  69 &  78 & 123 & 0.762                                  \cr
  8   &  78 &  96 & 112 & 176 & 0.768                                  \cr
  9   & 105 & 133 & 158 & 253 & 0.798                                  \cr
 \!\!\!10   & 138 & 182 & 212 & 342 & 0.795                                  \cr
 \!\!\!11   & 187 & 241 & 272 & 449 & 0.783                                  \cr
\hline
\end{tabular}$$
\caption{The minimal numbers of moves for the 4 different perfect tasks in \texttt{Path$_{4}$}}
\end{table}

The table prompts
\begin{question}
{\rm Is it the case that $|R_{4,1,n} \rightarrow R_{4,2,n}| < |R_{4,1,n} \rightarrow R_{4,3,n}|$ and $|R_{4,2,n} \rightarrow R_{4,3,n}| < |R_{4,1,n} \rightarrow R_{4,2,n}|$ \, for all $n \ge 3$?}
\end{question}
Both of these inequalities seem intuitively quite plausible.

\subsection{Upper bound for Path($4,n$)} \label{h4}
In this subsection we present the algorithm \texttt{FourMove} for moving a block $B$ of
size $n$ from peg 1 to peg 4 in \texttt{Path$_4$}, requiring no more than $1.6 \sqrt{n}
3^{\sqrt{2n}}$ moves. By Theorem~\ref{tp2}(a), this will imply Theorem~\ref{tp2}(b).
The description of \texttt{FourMove} is given in
Algorithm~\ref{alg041}.
\begin{algorithm}
\caption{\texttt{FourMove}($B$)}~\label{alg041}
\begin{algorithmic}
\STATE{\hspace{-12pt}/* $B$ is a block of disks, partitioned into $B_{s} \cup B_{l} \cup \{B_{\max}\}$, where $B_{s}$ is the set of smallest $\,\,\,$ */}
\STATE{\hspace{-12pt}/* disks, $B_{l}$ --- the larger, $\{B_{\max}\}$ --- the largest. \texttt{ThreeMove($B$,$s$,$d$,$a$)} returns the shortest $\,\!$*/}
\STATE{\hspace{-12pt}/* move-sequence for moving $B$ from $s$ to $d$ using $a$, referring to the graph induced by these $\,\!$ */}
\STATE{\hspace{-12pt}/* three vertices as \texttt{Path$_3$}. A move of disk~$n$ from peg~$i$ to peg~$i+1$ is
denoted by $t_{i, i+1,n}$. $\,\,\,\,\,\!$ */}
\STATE{\hspace{-12pt}/* The `*' denotes concatenation of move-sequences. $\qquad\qquad\qquad\qquad\qquad\qquad\qquad\qquad\qquad \!\! $*/}
\STATE {$M \leftarrow [\,]$ \qquad /* an empty sequence */}
\IF{$B \ne \emptyset$}
 \STATE{$n \leftarrow |B|$}
 \STATE{$m \leftarrow {\rm round}(\sqrt{2n})$}
 \STATE{$B_{s} \leftarrow [B_{\min},B_{\max}-m]$}
 \STATE{$B_{l} \leftarrow [B_{\max}-m+1,B_{\max}-1]$}
 \STATE{$M_{s} \leftarrow \texttt{FourMove}(B_{s})$}
 \STATE {\bf /* Spread: */}
 \STATE{$M \leftarrow M_{s} * $ \texttt{ThreeMove}($B_{l},1,3,2$) * $t_{1,2,n}$}
 \STATE {\bf /* Circular shift: */}
 \STATE{$M \leftarrow M * M_{s}^{-1} * $ \texttt{ThreeMove}($B_{l},3,4,2$) * $t_{2,3,n}$ * \texttt{ThreeMove}($B_{l},4,2,3$)}
 \STATE {\bf /* Accumulate: */}
 \STATE{$M \leftarrow M * $ $t_{3,4,n}$ * \texttt{ThreeMove}($B_{l},2,4,3$) * $M_{s}$}
\ENDIF
\STATE {return $M$}
\end{algorithmic}
\end{algorithm}

Prior to its main stages, \texttt{FourMove} partitions $B$ into three: a block containing the smallest disks, a block containing the larger ones, and a block containing a single disk -- the largest one. These blocks are denoted $B_{s}, B_{l}, \{B_{\max}\}$ respectively, with $m=|B_{l} \bigcup \{B_{\max}\}|$. Thus $B_{s} = [B_{\min},B_{\max}-m]$ and $B_{l} = [B_{\max}-m+1,B_{\max}-1]$.
In the three principal stages that follow: Spread, Circular shift and Accumulate, the moves are done based on these blocks. In Spread, $B_{s}$ is transferred to the farthest peg -- number~4, $B_{l}$ to peg~3, $\{B_{\max}\}$ to peg~2. In Accumulate, the opposite is done: these blocks are gathered on the destination peg. In-between the algorithm performs the Circular shift stage, whose role is to reverse the order of the blocks, so that it will be possible to perform the Accumulate stage. It is easy to verify that, as the execution of the algorithm terminates, all the blocks are legally gathered on the destination peg~$4$, as required.

The \texttt{ThreeMove} procedure in Algorithm~\ref{alg041}
produces move-sequences for $B_{l}$ using only
three (consecutive) pegs, which is exactly as moving it in
\texttt{Path$_3$}. To this end, we use the algorithm of
\cite{ScGrSm44}, which transfers a block in a minimal number of
moves between any two pegs in \texttt{Path$_3$}, requiring $3^{n}-1$
moves to transfer $n$ disks between the two farthest pegs, and half
that number of moves between neighboring pegs.

Denote by $T(n)$ the number of moves required by 
\texttt{FourMove} for a block of size $n$, and define $T(0) = 0$. Each of the
three recursive invocations of the algorithm \texttt{FourMove} with
$B_{s}$ requires $T(n-m)$ moves. Observe that, for a positive integer $n$, we have $1 \le m = $ round($\sqrt{2n}) \le n$. Employing the abovementioned results regarding the number of moves required by
\texttt{ThreeMove}, it is easy to see that the
total number of moves required by
$B_{s}$ is  $3(3^{m-1}-1) + \frac{1}{2}(3^{m-1}-1)$. Finally, $\{B_{\max}\}$ 
performs 3 moves.
Altogether, for $n \ge 1$ we have:
\begin{eqnarray}~\label{eqn01}
\begin{tabular}{rcl}
 $T(n)$ & = & $3 \cdot T(n-m) + \frac{7}{2} \cdot (3^{m-1}-1) + 3 $ \cr
        &   & \cr
        & = & $3\cdot T(n-m) + \frac{7}{6} \cdot 3^{m} - \frac{1}{2}$,
\end{tabular}
\end{eqnarray}
where $~m~=~{\rm round}(\sqrt{2n})$.

Next, we prove by induction that $T(n) < 1.6 \sqrt{n} \cdot
3^{\sqrt{2n}}$, implying the required result. For the induction basis, we note that the inequality has been verified manually for all values of $n \le 8$. Let $n \ge 9$, and assume that the inequality holds when $n$ is replaced by a smaller integer. To prove it for $n$, denote first $\beta = m - \sqrt{2n}$. (Clearly, $-\frac{1}{2} < \beta < \frac{1}{2}$.) Note that
\begin{eqnarray}~\label{eqn02}
\sqrt{1-\frac{m}{n}} ~<~ 1- \frac{1}{\sqrt{2n}}-\frac{2\beta+1}{4n},
\end{eqnarray}
which can be verified by squaring both sides of the inequality, noting that the right-hand side is positive for $n \ge 2$. Thus, by the induction hypothesis and (\ref{eqn02}),
\begin{eqnarray*}
T(n) & = & 3 \cdot T(n-m) + \frac{7}{6} \cdot 3^{m} - \frac{1}{2} \cr
& & \cr
     & < & 3 \cdot 1.6 \sqrt{n-m} \cdot 3^{\sqrt{2(n-m)}} +
\frac{7}{6} \cdot 3^{m} \cr
& & \cr
     & = & 3 \cdot 1.6 \sqrt{n} \sqrt{1-\frac{m}{n}} \cdot 3^{\sqrt{2n}
\sqrt{1-\frac{m}{n}}} + \frac{7}{6} \cdot 3^m \cr
& & \cr
     & < & 3 \cdot 1.6 \sqrt{n} \left(1- \frac{1}{\sqrt{2n}}
-\frac{2\beta+1}{4n}\right) 3^{\sqrt{2n} \left(1-
\frac{1}{\sqrt{2n}} -\frac{2\beta+1}{4n}\right)} +
\frac{7}{6} \cdot 3^{\sqrt{2n}+\beta} \cr
& & \cr
     & = & 3^{\sqrt{2n}} \left[1.6 \sqrt{n}
\left(1- \frac{1}{\sqrt{2n}} -\frac{2\beta+1}{4n}\right)
3^{-\frac{2\beta+1}{2\sqrt{2n}}} + \frac{7}{6} \cdot
3^{\beta}\right] \cr
& & \cr
     & = & 3^{\sqrt{2n}} \left[1.6 \sqrt{n} \left(1-
\frac{1}{\sqrt{2n}} -\frac{2\beta+1}{4n}\right) e^{-\frac{2\beta+1}{2\sqrt{2n}}\ln 3} + \frac{7}{6} \cdot
3^{\beta}\right].
\end{eqnarray*}
Notice that the term $-\frac{2\beta+1}{2\sqrt{2n}}\ln 3$ is negative for $\beta > -
\frac{1}{2}$. Hence, using the fact that $e^x \le 1+x +
\frac{1}{2}x^2$ for $x<0$, we obtain

$\begin{array}{llll}
\frac{T(n)}{1.6 \cdot 3^{\sqrt{2n}}} & < & & \sqrt{n} \left(1-
\frac{1}{\sqrt{2n}} -\frac{2\beta+1}{4n}\right) \left(1 -
\frac{2\beta+1}{2\sqrt{2n}}\ln 3 + \frac{(2\beta+1)^{2}\ln^{2}3}{16n}\right) + \frac{35}{48} \cdot 3^{\beta} \cr
& & & \cr
& = & & \left(\sqrt{n} - \frac{1}{\sqrt{2}} -
\frac{(2\beta+1)\ln 3}{2\sqrt{2}}\right) \cr
& & & \cr
& & + & \left(-\frac{2\beta+1}{4\sqrt{n}} + \frac{
(2\beta+1)\ln 3}{4\sqrt{n}} + \frac{(2\beta+1)^{2} \ln^{2}3}{16
\sqrt{n}}\right) \cr
& & & \cr
& & + & \left(\frac{
(2\beta+1)^2 \ln 3}{8\sqrt{2} \cdot n}
 - \frac{(2\beta+1)^{2}\ln^{2}3}{16
\sqrt{2}\cdot n} - \frac{(2\beta+1)^{3}\ln^{2}3}{64n \sqrt{n}} \right) + \frac{35}{48}
\cdot 3^{\beta} \cr


 &   &   & \cr
 & = & & \left(\sqrt{n} - \frac{1}{\sqrt{2}} -
\frac{(2\beta+1)\ln 3}{2\sqrt{2}}\right) \cr
& & & \cr
& & + & \left(-\frac{2\beta+1}{4\sqrt{n}} + \frac{
(2\beta+1)\ln 3}{4\sqrt{n}} + \frac{(2\beta+1)^{2}\ln^{2}3}{16
\sqrt{n}} + \frac{(2\beta+1)^{2}\ln^{2}3}{16\sqrt{2} \cdot n}\right) \cr
& & & \cr
& & + & \left( \frac{
(2\beta+1)^2 \ln 3}{8\sqrt{2} \cdot n}
 - \frac{(2\beta+1)^{2}\ln^{2}3}{8
\sqrt{2} \cdot n} - \frac{(2\beta+1)^{3}\ln^{2}3}{64n \sqrt{n}} \right) + \frac{35}{48}
\cdot 3^{\beta}.
\end{array}$\\

Observe that
$$\frac{(2\beta+1)^2 \ln 3}{8\sqrt{2} \cdot n}
 - \frac{(2\beta+1)^{2}\ln^{2}3}{8\sqrt{2}\cdot n} -
 \frac{(2\beta+1)^{3}\ln^{2}3}{64 \cdot n \sqrt{n}} < 0$$ \\
for $\beta > -\frac{1}{2}$, and therefore

$\begin{array}{llll}
\frac{T(n)}{1.6 \cdot 3^{\sqrt{2n}}} & < & & \left(\sqrt{n} -
\frac{1}{\sqrt{2}} - \frac{(2\beta+1) \ln 3}{2\sqrt{2}}\right) \cr
& & & \cr
& & + & \left(-\frac{2\beta+1}{4\sqrt{n}} + \frac{
(2\beta+1)\ln 3}{4\sqrt{n}} + \frac{(2\beta+1)^{2}\ln^{2}3}{16\sqrt{n}} + \frac{(2\beta+1)^{2}\ln^{2}3}{16 \sqrt{2}\cdot n}\right) + \frac{35}{48} \cdot 3^{\beta} \cr
& & & \cr
& = & & \left(\sqrt{n} - \frac{1}{\sqrt{2}} - \frac{
(2\beta+1)\ln 3}{2\sqrt{2}} \right) \cr
& & & \cr
& & + & \frac{2\beta +1}{16 \sqrt{n}}
\left(4(\ln 3 -1) + (2\beta +1)\ln^2 3 \left(1 + \frac{1}{\sqrt{2n}}\right) \right) + \frac{35}{48} \cdot
3^{\beta}.
\end{array}$ \\

Now for $\beta > -\frac{1}{2}$ the expression
\begin {eqnarray} \label{forimp}
 \frac{2\beta +1}{16 \sqrt{n}} \left(4(\ln 3 -1) + (2\beta +1)\ln^2 3 \left(1 + \frac{1}{\sqrt{2n}}\right) \right)
\end {eqnarray}
increases as a function of $\beta$ and decreases as a function of $n$. Hence its maximal value in the range $-\frac{1}{2} < \beta \le \frac{1}{2}, n \ge 9$, is obtained for $\beta=\frac{1}{2}, n=9$. Since its value at that point is less than $0.15$, we have


\begin {eqnarray} \label{eq:fi}
\frac{T(n)}{1.6 \cdot 3 ^{\sqrt{2n}}} &<& \sqrt{n} +
\left(-\frac{1}{\sqrt{2}} - \frac{(2\beta+1)\ln 3}{2\sqrt{2}} + 0.15 + \frac{35}{48} \cdot
3^{\beta}\right).
\end {eqnarray}
It is easy to verify that for $-\frac{1}{2} < \beta \le
\frac{1}{2}$,
\begin{eqnarray} \label{eq:fi2}
f(\beta) = -\frac{1}{\sqrt{2}} - \frac{(2\beta+1)\ln 3}{2\sqrt{2}} + 0.15 + \frac{35}{48} \cdot
3^{\beta} \end{eqnarray} is maximized at $\beta = \frac{1}{2}$, and
$f(\frac{1}{2}) < 0$. Consequently, the right-hand side of
(\ref{eq:fi}) is smaller than $\sqrt{n}$, and we are
done.

\subsubsection{\large A better upper bound for Path($4,n$)}
We reduced the problem of upper bounding Path($4,n$) to
the problem of upper bounding the following recurrence formula,
which might be of independent interest:
\begin{eqnarray*}
T(n)  & = & \min_{1\le m \le n} \left( 3 \cdot T(n-m) + \frac{7}{6} \cdot 3^{m} - \frac{1}{2} \right) \, .
\end{eqnarray*}
We obtained an upper bound of $1.6 \cdot \sqrt{n} \cdot
3^{\sqrt{2n}}$ for this recurrence formula, which is tight up to the
leading constant $1.6$.

One way to decrease the  constant $1.6$ is to show, using a computer
program, that a better upper bound holds for all values of $n \le
n'$, for some huge integer $n'$. This will serve as a significantly
more elaborate induction basis than the one that we use above (i.e.,
$n \le 8$), and consequently, it would suffice to prove the
induction step for $n > n'$ only. The maximum value of the function
$g(n,\beta)$ defined in (\ref{forimp}) for a huge integer $n > n'$
and $-\frac{1}{2} < \beta \le \frac{1}{2}$ is some tiny number
$\varepsilon = \varepsilon(n)$, and so, substituting $0.15$ with
$\varepsilon$ in (\ref{eq:fi2}) yields:
$$f'(\beta) = -\frac{1.6}{\sqrt{2}} - \frac{1.6 \cdot \ln 3 (2\beta+1)}{2\sqrt{2}} + 1.6 \cdot
\varepsilon + \frac{7}{6} \cdot 3^{\beta},$$ with $f(\beta)-f'(\beta) =
1.6(0.15-\varepsilon) \approx 1.6 \cdot 0.15$. The difference
$1.6(0.15-\varepsilon)$, between $f(\beta)$ and $f'(\beta)$, for a sufficiently small $\varepsilon$, enables us to
decrease the leading constant $1.6$ to approach $1.365$, yielding an
upper bound that approaches $1.365 \cdot \sqrt{n} \cdot
3^{\sqrt{2n}}$.

A more involved method for decreasing the above constant is to
choose another value for $m$. For technical convenience, we fixed $m
= {\rm round}(\sqrt{2n})$, but this choice of $m$ is inherently
suboptimal. By following the method outlined in the previous
paragraph, and setting $m =
\left\lfloor\sqrt{2n}+\left(1-\frac{1}{\ln 3}\right) \right\rfloor$,
one can achieve a constant that approaches $1.105$, yielding an
upper bound of approximately $1.105 \cdot \sqrt{n} \cdot
3^{\sqrt{2n}}$.

\renewcommand{\arraystretch}{2.4}

\section{Proof of Theorem \ref{tp1}}~\label{sec5}
The proof of Theorem \ref{tp1} is organized as follows. Generally, we would like to show how one can move a column of $n$ disks from any source peg~$s$ to any destination peg~$d$ such that the number of moves is bounded above as the theorem states. For simplicity, we start by presenting an algorithm for the case where $s=1, d=h$. This will be done in Section~\ref{sec51}. Then we present an algorithm for the general case (Section~\ref{sec52}). We note that, in fact, the first algorithm does employ the second. An important point in both cases is a partitioning of the set of disks to blocks, which will be discussed in Section~\ref{sec53}. Time analysis of the two algorithms will be provided in Section~\ref{sec54}.


\subsection{Moving disks between the farthermost pegs}~\label{sec51}
Here we present \texttt{FarthestMove} (Algorithm~\ref{alg4333}), designed to move a block $B$ of $n$ disks between the two farthest pegs in \texttt{Path$_h$}, where $h \ge 3$.

We partition $B$ in some way to blocks $B_1(h,B), B_2(h,B), \ldots, B_{h-1}(h,B)$ of disks. Whenever $h$ and $B$ are implied by the context, we write $B_i$ instead of $B_i(h,B)$. The block $B_1$ consists of the
smallest $\widetilde{n}_1$ disks $1, 2, \ldots, \widetilde{n}_1$, the block $B_2$ --- of the $\widetilde{n}_2$ next smallest disks $\widetilde{n}_1 + 1, \widetilde{n}_1 + 2, \ldots, \widetilde{n}_1 + \widetilde{n}_2$, and so forth. Similarly to the shorthand used when denoting blocks, we may write $\widetilde{n}_i$ (with a possible superscript) instead of $\widetilde{n}_i(h,n)$. For any $i \in [h-1]$, let $B(i) = \bigcup_{j=i}^{h-1}
B_j$ and $n(i) = |B(i)| = \Sigma_{j=i}^{h-1}
\widetilde{n}_j$, where $\widetilde{n}_j = |B_j|$. (Note that $B(1) = B$ and $n(1) = n$.)

The determination of the sizes $\widetilde{n}_i$ is crucial for the number of moves the algorithm makes, and will be explained later. However, for the algorithm to work correctly, it is only required for $B_{h-1}$ to consist of the single disk $B_{\max}$ --- the largest.
The algorithm consists of three phases (see Figure \ref{fig:far} for an illustration):
\begin{itemize}
\item {\bf Spread:} Move the $h-2$ $(=d-s-1)$ first
blocks $B_1,\ldots,B_{h-2}$ from the source peg~$s$ to pegs
$d, \ldots, d-h+3$, respectively. It consists of $h-2$
iterations. At the $j$-th iteration, $j \in [h-2]$, block $B_j$
is (recursively) moved from $s$ to $d-j+1$,
using the set $[1,d-j+1]$ of available pegs. (Note that the $1$-disk block $B_{h-1}$ has not been moved from $s$ to $s+1$. It is more convenient for us to view this move as the first move of the next stage.)

\item {\bf Reverse:} The role of this phase is to reverse the positions of
the $h-1$ blocks on the $h$ pegs, i.e., a block residing, at the beginning of this phase, on peg
$s+i-1$ reaches, at the end of the phase, its reflected position --- peg $d-i+1$.
The phase starts by moving the last block $B_{h-1}$ from $s$ to $s+1$. Then, $h-2$ rounds are carried out, each of which brings the next larger block to its reflected position.
The following highlights the way each round $j$ achieves its goal:
\begin{itemize}
\item Before this round, blocks $B_1, \ldots, B_{j-1}$ are on pegs $s, \ldots, s+j-2$, respectively; peg $s+j-1$ is vacant;
blocks $B_j, \ldots, B_{h-1}$ are on pegs $d, \ldots, s+j$, respectively.
\item Block $B_j$ is moved from $d$ to $s+j-1$.
\item Blocks $B_{j+1}, \ldots, B_{h-1}$ are each shifted one peg to the right.
\item At the end of the round, blocks $B_1, \ldots, B_{j}$ are on pegs $s, \ldots, s+j-1$, respectively; peg $s+j$ is vacant; blocks $B_{j+1}, \ldots, B_{h-1}$ are on pegs $d, \ldots, s+j+1$, respectively.
\end{itemize}
Thus, as a result of this phase, block $B_{h-1}$
is moved from $s$ to $d$ and, for each $j \in [h-2]$,
block $B_j$ is moved from peg $d-j+1$ to the reflected position,
namely, peg $s+j-1$.

\item {\bf Accumulate:} The role of this phase is symmetrical to that of {\bf
Spread}, i.e., to move the $h-2$ first blocks $B_{h-2},\ldots,B_{1}$
from pegs $d-2,\ldots,s$, respectively, to $d$. Similarly, it consists of $h-2$ iterations, where at the $j$-th iteration block $B_{h-1-j}$ is
moved from $s+h-2-j$ to $d$ using the set
$[s+h-2-j,d]$ of available pegs.
\end{itemize}
It is easy to verify that, as the execution of the algorithm
terminates, all the blocks are legally gathered on $d$, as required.
\begin{figure*}[htp]
\begin{center}
\begin{minipage}{\textwidth} 
\begin{center}
\setlength{\epsfxsize}{3.7in} \epsfbox{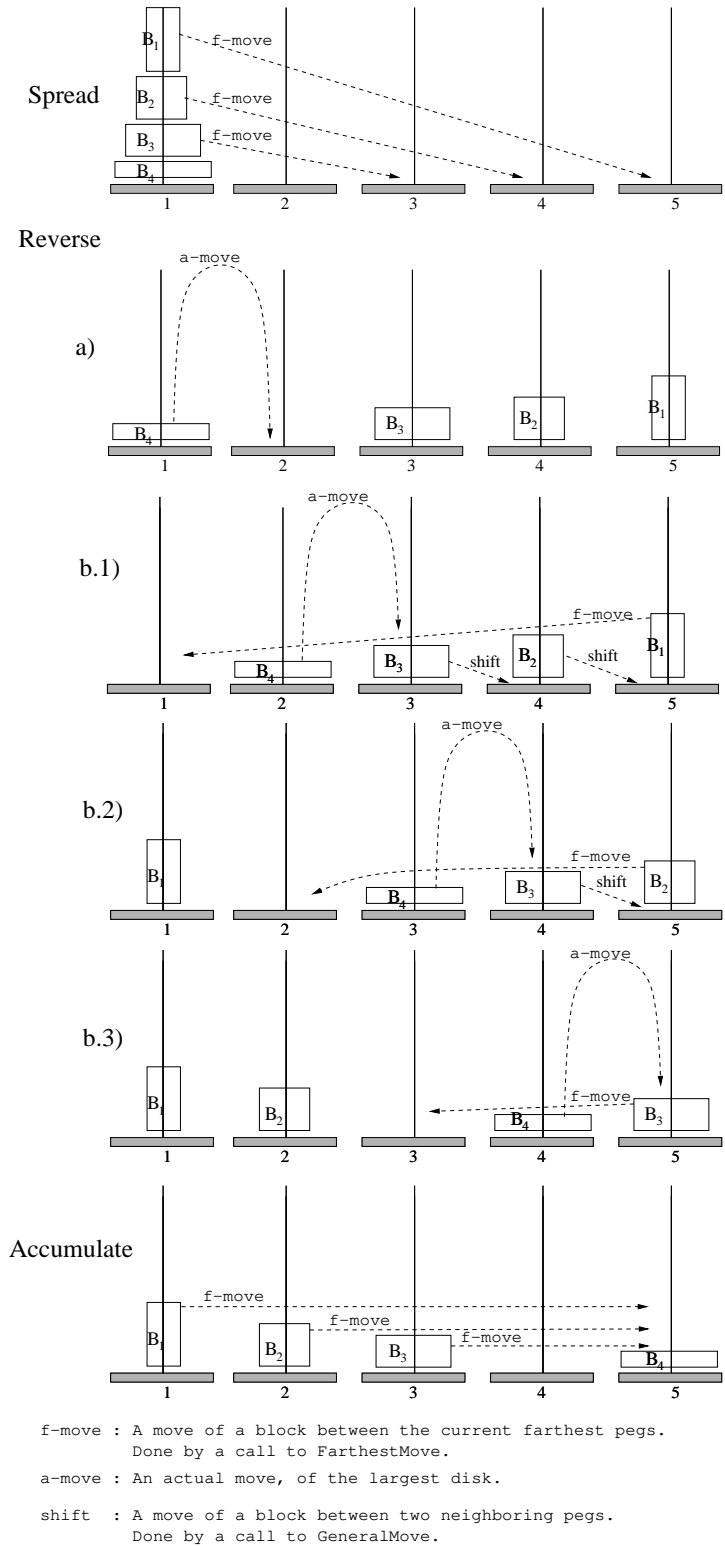}
\end{center}
\end{minipage}
\caption[]{ \label{fig:far} \sf Main steps of \texttt{FarthestMove} for \texttt{Path$_5$}, with $B =
\bigcup_{i=1}^{4} B_i$, $s = 1$ and $d = 5$.}
\end{center}
\end{figure*}
The formal description of the algorithm \texttt{FarthestMove} is given in
Algorithm \ref{alg4333}.

\begin{algorithm}
\caption{\texttt{FarthestMove}($B,s,d)$}~\label{alg4333}
\begin{algorithmic}
\STATE{\hspace{-12pt}/* The procedure moves a block $B$ from the leftmost peg $s$ to the rightmost peg $d$. Prior $\quad\,\,\,$ */}
\STATE{\hspace{-12pt}/* to its main stages, the block is partitioned into $h-1$ blocks which are treated as `atomic $\,\,\,\!$ */}
\STATE{\hspace{-12pt}/* units'. At the first of the main stages, these blocks are spread along the pegs; at the $\qquad\,\,$ */}
\STATE{\hspace{-12pt}/* second -- their order is reversed; at the third stage they are accumulated on the destination $\!$*/}
\STATE{\hspace{-12pt}/* peg. The procedure requires that $s < d$. If $d-s=1$, then $|B| \le 1$. The `*' denotes \qquad\,\,\,\,\,\, */}
\STATE{\hspace{-12pt}/* concatenation of move-sequences. $\qquad\qquad\qquad\qquad\qquad\qquad\qquad\qquad\qquad \qquad\qquad\qquad \,\,\,\,\,\,\, $*/}
\STATE {$T \leftarrow [\,]$ \qquad /* initializing the result sequence */}
\IF{$B \ne \emptyset$}
\STATE{$h \leftarrow d-s +1$}
\STATE{$(B_1,\ldots,B_{h-1}) = \texttt{Partition}(h,B)$} /* Algorithm~\ref{alg42} below */
\STATE{}
\STATE{\bf /* Spread: */}
\FOR{$j \leftarrow 1$ to $h-2$}
\STATE{/* At each step, the next block moves to the farthest available peg. */}
\STATE{$T \leftarrow T * \texttt{FarthestMove}(B_j,s,d-j+1)$}
\ENDFOR
\STATE{}
\STATE{\bf /* Reverse: */}
\STATE{$T \leftarrow T * t_{\mathsmaller{B}_{\max},s,s+1}$} \hskip12pt /* Moving the largest disk once a peg to the right. */
\STATE {$M \leftarrow [\,]$ \qquad\qquad\qquad /* Initializing the temporary move-sequence. */}
\FOR{$j \leftarrow 1$ to $h-2$}
\STATE{/* Block $B_j$ moves to the peg on which it will stay for the rest of this phase. */}
\STATE{$M_f \leftarrow \texttt{FarthestMove}(B_j,s+j-1,d)$}
\STATE $M \leftarrow M * M_f^{-1}$
\FOR{$i \leftarrow j+1$ to $h-2$}
\STATE{/* Each block whose index is higher than $j$ is shifted one peg to the right. */}
\STATE{$M \leftarrow M * \texttt{GeneralMove}(B_i,d+j-i,d+j+1-i,[s+j,d+j+1-i])$}
\ENDFOR
\STATE{$T \leftarrow T * M * t_{\mathsmaller{B}_{\max},j+1,j+2}$} \hskip6pt/* Moving the largest disk a peg to the right. */
\ENDFOR
\STATE{}
\STATE{\bf /* Accumulate: */}
\FOR{$j \leftarrow 1$ to $h-2$}
\STATE{/* At each step, the next block is gathered on the destination peg. */}
\STATE{$ T \leftarrow T * \texttt{FarthestMove}(B_{h-1-j},s+h-2-j,d)$}
\ENDFOR
\ENDIF
\STATE {return $T$}
\end{algorithmic}
\end{algorithm}


\subsection{Moving disks between any pegs}~\label{sec52}
The general algorithm for moving a block of disks, between
any two pegs $s$ and $d$, in \texttt{Path$_h$}, is presented here.
For convenience we assume that $s < d$. This does not effect the generality of the algorithm since, as was mentioned in Section~\ref{sec2}, if $M$ is a solution of $R_{h,s,n} \rightarrow R_{h,d,n}$, then $M^{-1}$ is a solution of $R_{h,d,n} \rightarrow R_{h,s,n}$.

The issue of partitioning the disk set is handled exactly as it was done in
\texttt{FarthestMove}. Algorithm \texttt{GeneralMove} consists of five phases: two spread phases, a phase in which the remainder disks are moved, and two accumulate phases. The set of available pegs is denoted by $A$, and its smallest and largest pegs by $A_{\min}$ and $A_{\max}$, respectively.
\begin{itemize}
\item {\bf LeftSpread:} In this phase the $s-A_{\min}$ first blocks
$B_1,\ldots,B_{s-\mathsmaller{A}_{\min}}$ are taken from peg $s$ to pegs
$A_{\min},\ldots,s-1$, respectively. It consists of $s-A_{\min}$ iterations. At the $j$-th iteration, $1 \le j \le s-A_{\min}$, block $B_j$ is (recursively) moved from $s$ to $A_{\min}+j-1$ using the set
$[A_{\min}+j-1,A_{\max}]$ of available pegs.

\item {\bf RightSpread:} Here, the $A_{\max}-d$ next blocks
are taken, from peg $s$ to pegs $A_{\max},\ldots,d+1$,
respectively. At each iteration $j$, where $1 \le j \le A_{\max}-d$, block $B_{s-\mathsmaller{A}_{\min}+j}$ is moved from $s$ to $A_{\max}-j+1$, using $[s,A_{\max}-j+1]$. Since at each iteration the source and destination are at the opposite ends of the currently available set of free pegs, the move is done using algorithm \texttt{FarthestMove}.

\item {\bf MoveRemainder:} In this phase, the remaining $B(d-s)$ blocks
are moved from $s$ to $d$. Since, as before, the source and destination are at the opposite sides of the set $[s,d]$ of available pegs, this is done by algorithm \texttt{FarthestMove}.

\item {\bf LeftAccumulate:} The role of this phase is symmetrical to that of
{\bf RightSpread}, that is, move $B_{s-A_{\min}+1},\ldots,B_{s+|A|-1-d}$ from  $A_{\max},\ldots,d+1$ to $d$,
respectively. It consists of $A_{\max}-d$ iterations, where
at iteration $j$, block $B_{s +|A| - d - j}$ is moved from $d+j$ to $d$ using $[s,d+j]$. Unlike {\bf RightSpread}, the moves made in this phase are not between the two farthest available pegs.

\item {\bf RightAccumulate:} This phase is symmetrical to
{\bf LeftSpread}, consisting of $s-A_{\min}$ iterations where, at the $j$-th iteration,
$B_{s-j}$ is moved from peg $s - j$ to peg $d$, using $[s - j,A_{\max}]$.
\end {itemize}
It is easy to verify that, as the algorithm
terminates, all the blocks are legally gathered on the destination
peg $d$, as required (see
Figure \ref{fig:gen} for an illustration). The correctness proof is omitted.
The formal description of \texttt{GeneralMove} is given in
Algorithm \ref{alg4222}.
Note that, if the source and destination pegs are at the opposite sides of $A$, then \texttt{GeneralMove} does the same as \texttt{FarthestMove}.
\begin{figure*}[htp]
\begin{center}
\begin{minipage}{\textwidth} 
\begin{center}
\setlength{\epsfxsize}{4.2in} \epsfbox{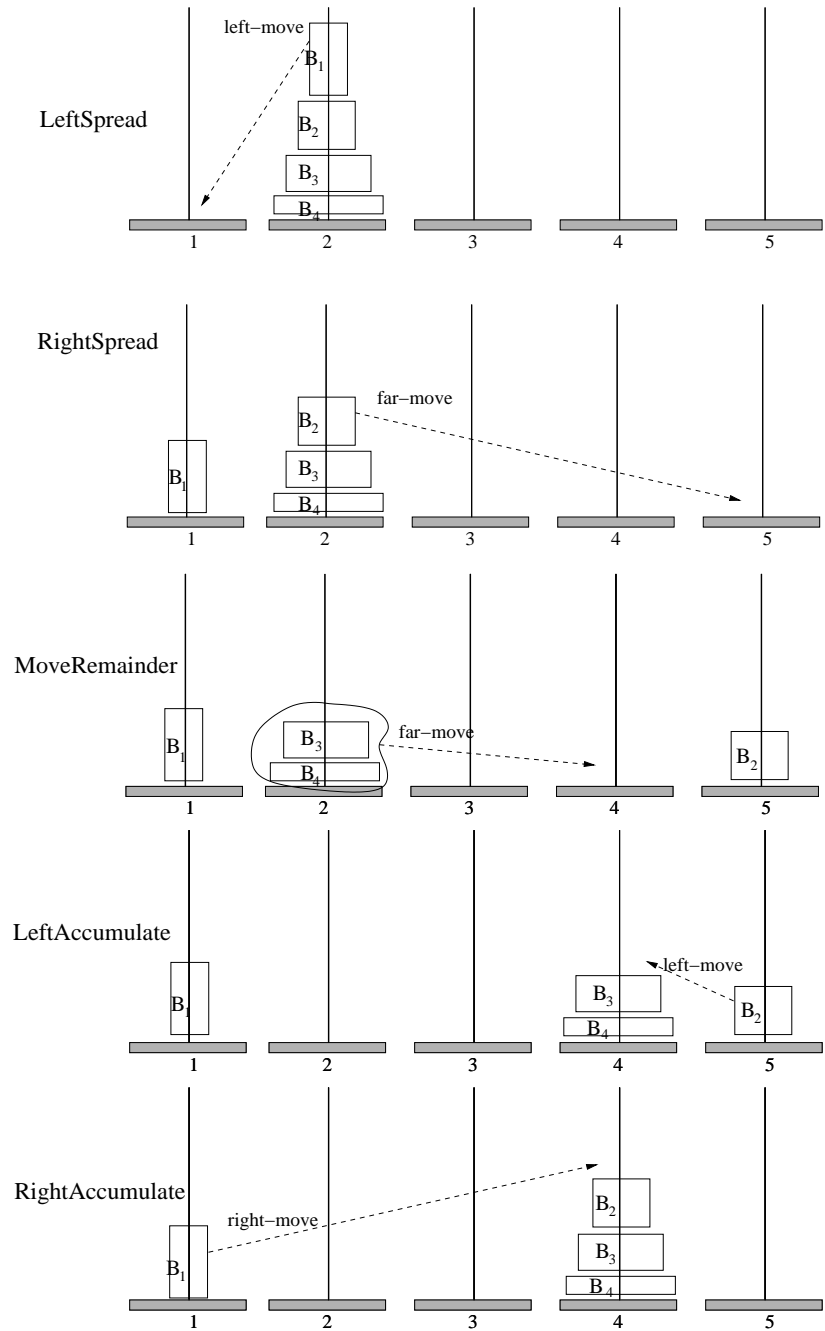}
\end{center}
\end{minipage}
\caption[]{ \label{fig:gen} \sf An illustration of the execution of
the algorithm \texttt{GeneralMove} for \texttt{Path$_5$}, with $B =
\bigcup_{i=1}^{4} B_i$, $s = 2$, $d = 4$ and $A = [5]$.}
\end{center}
\end{figure*}

\begin{algorithm}
\caption{$\rm\texttt{GeneralMove}(B,s,d,A)$}~\label{alg4222}
\begin{algorithmic}
\STATE{\hspace{-12pt}/* The procedure moves a block $B$ from any peg $s$ to any other peg $d$. The partitioning is as in $\,\,\,$ */}
\STATE{\hspace{-12pt}/* \texttt{FarthestMove}. Then some blocks are spread to the left of $s$ and to the right of $d$; next -- the $\,\,$ */}
\STATE{\hspace{-12pt}/* remaining blocks are moved; eventually, the smallest disks are accumulated on the destination $\!$*/}
\STATE{\hspace{-12pt}/* peg. The procedure requires that $s,d \in A$, $s < d$. If $d-s=1$, then $|B| \le 1$. The `*' denotes \! */}
\STATE{\hspace{-12pt}/* concatenation of move-sequences. $\qquad\qquad\qquad\qquad\qquad\qquad\qquad\qquad \qquad\qquad\qquad\qquad\qquad \, $*/}
\STATE {$T \leftarrow [\,]$ \qquad /* an empty sequence */}
\IF{$B \ne \emptyset$}
 \STATE{$(B_1,\ldots,B_{|A|-1}) = {\rm\texttt{Partition}}(|A|,B)$} /* Algorithm~\ref{alg42} below */
 \STATE{}
 \STATE{\hspace{-12pt}\bf /* LeftSpread: */}
 \FOR{$j \leftarrow 1$ to $s-A_{\min}$}
  \STATE{/* At each step, the next block moves to the farthest available peg on the left side. Since */}
  \STATE{/* the source is not necessarily at the rightmost available peg, this is done by \texttt{GeneralMove}.*/}
  \STATE $f_n \leftarrow A_{\min}+j-1$
  \STATE{$M \leftarrow \texttt{GeneralMove}(B_j,f_n,s,[f_n,A_{\max}])$}
  \STATE $T \leftarrow T * M^{-1}$
 \ENDFOR
 \STATE{}
 \STATE{\hspace{-12pt}\bf /* RightSpread: */}
 \FOR{$j \leftarrow 1$ to $A_{\max}-d$}
  \STATE{/* At each step, the next block moves to the farthest available peg on the right side. */}
  \STATE{/* Since the source is at the leftmost available peg, this is done by \texttt{FarthestMove}.*/}
  \STATE $f_x \leftarrow A_{\max}-j+1$
  \STATE{$T \leftarrow T * \texttt{FarthestMove}(B_{s-\mathsmaller{A}_{\min}+j},s,f_x$)}
 \ENDFOR
 \STATE{}
 \STATE{\hspace{-12pt}\bf /* MoveRemainder: */}
 \STATE{$T \leftarrow T * \texttt{FarthestMove}(B(|A|-d+s),s,d)$}
 \STATE{}
 \STATE{\hspace{-12pt}\bf /* LeftAccumulate: */}
 \FOR{$j \leftarrow 1$ to $A_{\max}-d$}
  \STATE $f_x \leftarrow d+j$
  \STATE $M \leftarrow \texttt{GeneralMove}(B_{|A|+s-d-j},d,f_x,[s,f_x])$
  \STATE $T \leftarrow T * M^{-1}$
 \ENDFOR
 \STATE{}
 \STATE{\hspace{-12pt}\bf /* RightAccumulate: */}
 \FOR{$j \leftarrow 1$ to $s-A_{\min}$}
  \STATE $f_n \leftarrow s-j$
  \STATE{$T \leftarrow T * \texttt{GeneralMove}(B_{s-j},f_n,d,[f_n,A_{\max}])$}
 \ENDFOR
\ENDIF
\STATE {return $T$}
\end{algorithmic}
\end{algorithm}


\subsection{Partitioning the disks into blocks}~\label{sec53}
In this section we discuss how to set the sizes of the blocks such that the number of moves will be relatively low.
The general idea is to view the $h-1$ blocks as `atomic' units, each
occupying a single peg (except for when it is moved). During
the process of moving a block $B_i$ from one peg $s$ to another
peg $d$, the other blocks stay intact. Furthermore, the pegs
used by disks from $B_i$ during this process form an interval of
contiguous integers, contained in the set of pegs
available to this end, namely, the inclusion-wise maximal
interval of pegs not occupied by any of the blocks
$B_{1},\ldots,B_{i-1}$.

To move a block between pegs efficiently,  all available
pegs should usually be in use. More specifically, during the process of
moving a sufficiently large block $B_i$, all of the
available pegs are used. Furthermore, the algorithm allocates
precisely $h-i+1$ pegs to this end. This suggests that, in order to
perform efficiently, the sizes of the $h-1$ blocks should satisfy
$\widetilde{n}_1 \ge
 \ldots \ge \widetilde{n}_{h-1}=1$ (assuming $n$ is sufficiently large).

\subsubsection{\large The Partition procedure}
In this section we present \texttt{Partition} --- the procedure for
partitioning a block $B$ into the $h-1$ blocks $(B_{1}, B_{2}, \ldots, B_{h-1})$.
We start by presenting an auxiliary function \texttt{Remainder},
which, for each stage~$j$, provides the total number of disks to be assigned to the latter blocks --
$(B_{j+1}, B_{j+2}, \ldots, B_{h-1})$. The definition of this function is given in Algorithm~\ref{alg413}.

\begin{algorithm}
\caption{\texttt{Remainder($h,n$)}}~\label{alg413}
\begin{algorithmic}
\STATE{/* Determines the size of the next set, by calculating the number of `larger' disks. */}
\STATE{/* It is assumed that $h \ge 3$ and $n \ge 1$. */}
\IF{$n < h$}
\STATE{return $\max\{n-1,1\}$}
\ELSE
\IF{$h=4$}
\STATE{return $\min\{n,{\rm round}(\sqrt{2n})\}$}
\ELSE
\STATE{$\alpha \leftarrow \frac{h-3}{h-2}$}
\STATE{return $\min\left\{n,\left\lceil\frac{((h-2)!)^{\alpha}}{(h-3)!}n^{\alpha}\right\rceil\right\}$}
\ENDIF
\ENDIF
\end{algorithmic}
\end{algorithm}

\begin{lemma}\label{remainder02}
For any integers $h \ge 3$ and $n \ge 1$, we have $1 \le$
{\rm\texttt{Remainder}($h,n$)} $\le n$. Furthermore,
\begin {itemize}
\item If $h=3$, then {\rm\texttt{Remainder}($h,n$)} = $1$.
\item If $h \ge 4$ and $1 < n < h$, then {\rm\texttt{Remainder}($h,n$)} = $n-1$.
\item If $h \ge 4$ and $n \ge h$, then {\rm\texttt{Remainder}($h,n$)} $\ge 2$.
\end {itemize}
\end{lemma}
The proof is straightforward.

The formal description of the procedure \texttt{Partition} is given in
Algorithm \ref{alg42}.
\begin{algorithm}
\caption{\texttt{Partition}($h,B$)}~\label{alg42}
\begin{algorithmic}
\STATE{/* Returns a partition of the $n$ disks into $h-1$ blocks of consecutive disks. */}
\STATE{/* It is assumed that $h \ge 2$ and $B$ is a non-empty block of disks. */}
\FOR{$j \leftarrow 1$ to $h-2$}
\STATE{$n_{j} \leftarrow B_{\max}+1-B_{\min}$}
\STATE{${m}_j \leftarrow $ \rm\texttt{Remainder}($h-j+1,n_j$)}
\STATE{$\widetilde {n}_j \leftarrow n_j - m_j$}
\STATE{$B_j \leftarrow [B_{\min},B_{\min}+\widetilde{n}_j -1]$ \qquad /* if $\widetilde{n}_j = 0$, then $B_j=\emptyset$ */}
\STATE{$B \leftarrow B - B_j$}
\ENDFOR
\STATE{/* B is now a singleton, so that $B_{h-1} = \{B_{\max}\}$ */}
\STATE{$B_{h-1} \leftarrow B$}
\STATE{return $(B_1,\ldots,B_{h-1})$}
\end{algorithmic}
\end{algorithm}

We argue that \texttt{Partition} is well-defined. To prove
this, it suffices to show that at each of the $h-2$ invocations of
\texttt{Remainder}($h-j+1$,$n_j$), $1
\le j \le h-2$, we have $h-j+1 \ge 3$ and $n_j \ge 1$.
The first of these inequalities follows from the fact that $j \le h-2$. Now observe that $n_1
= n \ge 1$, and $n_{j+1} = m_j = $ \texttt{Remainder}($h-j+1,n_j$). Hence, by
Lemma~\ref{remainder02}, a simple inductive argument yields
\begin{equation} \label{eq:s} n_j \ge 1, \qquad 1 \le j \le h-2,
\end{equation}
and we are done.

In the following lemma, whose proof is straightforward, we collect for later reference a few properties of the partition $(B_1,\ldots,B_{h-1})$.
\begin{lemma}\label{partition03}
The tuple $(B_1,\ldots,B_{h-1})$ is a partition of $B$ into
blocks, satisfying:
\begin {itemize}
\item $B_{h-1} = \{B_{\max}\}$.
\item $|B_j| = \widetilde{n}_j \le n-1$ for each $j \in [h-2]$.
\item Each non-empty block is lighter than all subsequent non-empty blocks in the partition.
\end{itemize}
\end {lemma}

It is easy to verify that, for a pair of indices $i \in [h-1]$ and $j \in
[h-i]$,
$$
 B_{j+i-1}(h,B) = B_{j}(h-i+1,B(i)),
$$
or, equivalently:
 \begin{lemma}\label{blocksize04}
 For any integers $1 \le i \le h-1$ and $1 \le j \le h-i$:
 $$\widetilde{n}_{j+i-1}(h,n) = \widetilde{n}_{j}(h-i+1,n(i)).$$
 \end{lemma}

\subsection{\texttt{FarthestMove} \texttt{versus} \texttt{GeneralMove}} \label{sec54}
We assume without loss of generality that $A=[1,h]$. For any integers $h \ge 3,n \ge 0,s$
and $d$, such that $1 \le s < d \le h$, we denote by $G_{s
\rightarrow d}(h,n)$ the number of moves
required by \texttt{GeneralMove} to
move a block of size $n$ from peg $s$ to peg $d$
using $A$. Similarly, we denote by $F(h,n)$ the number of moves required by \texttt{FarthestMove} to move
such a block from peg~1 to peg~$h$. (Note that $F = G_{1 \rightarrow h}$.)

It is easy to verify that, for $h=3$, the algorithm \texttt{GeneralMove}
works exactly as does the classical algorithm of \cite{ScGrSm44}. In
particular, it requires $3^n-1$ moves to transfer $n$ disks between
the two farthest pegs in \texttt{Path$_3$}, and $\frac{3^n-1}{2}$ moves to transfer them between neighboring pegs, yielding:
\begin {lemma}\label{classical05}
For any non-negative integer $n$: $$G_{1 \rightarrow 2}(3,n) =
\frac{1}{2}F(3,n).$$
\end {lemma}

\subsubsection{\large \texttt{Initial steps in the analysis of} \texttt{GeneralMove}}~\label{sec541}
In this section we analyze the algorithm \texttt{GeneralMove} for moving a
block $B$ in \texttt{Path}$_{h}$, $h \ge 3$, from peg $s$ to peg $d$,
$s < d$, using the set $A = [1,h]$ of available pegs. Let $h' =
s + h - d$. (Note that $h' < h$.)

Consider an index $j \in [s-1]$. At phase {\bf LeftSpread}, a {\bf
left-move} of block $B_j$ from peg $s$ to peg $j$ using $h-j+1$
available pegs is performed, requiring $G_{j \rightarrow s}(h-j+1,\widetilde{n}_j)$ moves. Similarly, at phase {\bf
RightAccumulate}, a {\bf right-move} of block $B_j$ from peg $j$ to
peg $d$ using $h-j+1$ available pegs is performed, requiring
$G_{j \rightarrow d}(h-j+1,\widetilde{n}_j)$ moves.

 Consider now an index $j \in [s,h'-1]$.
At phase {\bf RightSpread}, a {\bf far-move} of block $B_j$ from peg
$s$ to peg $s+h-j$ using $h-j+1$ available pegs is performed,
requiring $F(h-j+1,\widetilde{n}_j)$ moves. At phase {\bf
LeftAccumulate}, a {\bf left-move} of block $B_j$ from peg $s+h-j$
to peg $d$ using $h-j+1$ available pegs is performed, requiring
$G_{d \rightarrow s+h-j}(h-j+1,\widetilde{n}_j)$ moves.

The remainder $B(h')$ of blocks is moved in phase {\bf
MoveRemainder}, using a {\bf far-move} from peg $s$ to peg $d$,
which requires $F(h-h'+1,n(h'))$ moves.

The discussion above implies
\begin{lemma}~\label{general06}
\begin{eqnarray*}
\begin{array}{llll}
 G_{s \rightarrow d}(h,n) & = & &
 \displaystyle{\sum_{j=1}^{s-1} \left[G_{j \rightarrow s}(h-j+1,\widetilde{n}_j) + G_{j \rightarrow d}(h-j+1,\widetilde{n}_j)
 \right]} \cr
 & & + & \displaystyle{\sum_{j=s}^{h'-1} \left[F(h-j+1,\widetilde{n}_j) + G_{d \rightarrow
 s+h-j}(h-j+1,\widetilde{n}_j)\right]} \cr
 & & + & F(h-h'+1,n(h')).
\end{array}
\end{eqnarray*}
\end{lemma}

\subsubsection{\large \texttt{Initial steps in the analysis of} \texttt{FarthestMove}}~\label{sec542}
In this section we analyze the algorithm \texttt{FarthestMove} for moving a
block $B$ from peg 1 to peg $h$ in \texttt{Path}$_{h}$, $h \ge 3$, using
the set $A=[1,h]$ of available pegs.

First, observe that the last block $B_{h-1}$, namely disk $B_{\max}$, performs $h-1$ moves.

Consider an index $j \in
[h-2]$. At each of the phases {\bf Spread}, {\bf Reverse}, and {\bf
Accumulate}, a {\bf far-move} of block $B_j$ with $h-j+1$ free
pegs is performed, requiring a total of $3F(h-j+1,\widetilde{n}_j)$ moves. Also, $j-1$ {\bf shifts} of block
$B_j$ with $h-j+1$ free pegs are performed at phase {\bf
Reverse}, requiring altogether $(j-1)G_{1 \rightarrow
2}(h-j+1,\widetilde{n}_j)$ moves.

For $1 \le i \le h-1$, denote by $F(h,n)|_{n(i)}$ the number of moves of the $n(i)$
largest disks in the course of performing the algorithm \texttt{FarthestMove}. The explanation in the preceding paragraph yields

\begin{lemma}\label{farthest07}
For any integers $h \ge 2$, $n \ge 1$, and $1 \le i \le h-1$,
\begin {itemize}
\item $F(h,n)|_{n(i)} =
 \displaystyle{\sum_{j=i}^{h-2} \left[ 3F(h-j+1,\widetilde{n}_j)
+ (j-1)G_{1 \rightarrow 2}(h-j+1,\widetilde{n}_j) \right] +
h-1}.$
\item $F(h,n) =
 \displaystyle{\sum_{j=1}^{i-1} \left[ 3F(h-j+1,\widetilde{n}_j)
+ (j-1)G_{1 \rightarrow 2}(h-j+1,\widetilde{n}_j) \right] +
F(h,n)|_{n(i)}}.$
\end {itemize}
\end{lemma}

For the subsequent lemmas we put $m=h-k+1$, for $1 \le k \le h-1$.
\begin{lemma}\label{fnh08}
For any integers $h \ge 2$, $n\ge1$, and $1 \le k \le h-1$,
$$F(h,n)|_{n(k)} ~=~ F(m,n(k)) + (k-1)
\sum_{j=k}^{h-1} G_{1 \rightarrow 2}(h-j+1,\widetilde{n}_j).$$
\end {lemma}
\proof By Lemma~\ref{farthest07} in
the particular case $i=1$, \\

$\begin{array}{llll}
 F(m,n(k)) & = & & \displaystyle{\sum_{j=1}^{m-2}} [3F(m-j+1,\widetilde{n}_j(m,n(k))) \cr
           &   &+&  (j-1)G_{1 \rightarrow 2}(m-j+1,\widetilde{n}_j(m,n(k)))] + h-k.
\end{array}$
\\

By Lemma~\ref{blocksize04}, we obtain
\begin{eqnarray}\label{eq:first}
\begin{array}{llll}
 F(m,n(k)) & = &   & \displaystyle{\sum_{j=1}^{m-2}} [3F(m-j+1,\widetilde{n}_{j+k-1}) \cr
           &   & + &  (j-1) G_{1 \rightarrow 2}(m-j+1,\widetilde{n}_{j+k-1})] ~+~ h-k \cr
           & = &   & \displaystyle{\sum_{j=k}^{h-2}} [3F(h-j+1,\widetilde{n}_{j})+(j-k)G_{1 \rightarrow 2}(h-j+1,\widetilde{n}_j)] + h-k.
\end{array}
\end{eqnarray}


Observe that $G_{1 \rightarrow 2}(2,\widetilde{n}_{h-1}) = 1$. Thus,
 by Lemma~\ref{farthest07}, the right-hand side of
(\ref{eq:first}) reduces to
\begin{eqnarray*} F(h,n)|_{n(k)} - (k-1)\sum_{j=k}^{h-1} G_{1 \rightarrow 2}(h-j+1,\widetilde{n}_j),
\end{eqnarray*} and we are done. \QED

Lemma~\ref{farthest07} and Lemma \ref{fnh08} imply
\begin{corollary} \label{c:farh}
For any integers $h \ge 2, n \ge 1$, and $1 \le k \le h-1$,
\begin {eqnarray*}
\begin{array}{llll}
 F(h,n) & = & & \displaystyle{\sum_{j=1}^{k-1} [3F(h-j+1,\widetilde{n}_j)
 +(j-1)G_{1 \rightarrow 2}(h-j+1,\widetilde{n}_j)]} \cr
        &   & + & \displaystyle{F(m,n(k)) + (k-1)\sum_{j=k}^{h-1} G_{1 \rightarrow 2}(h-j+1,\widetilde{n}_j)}.
\end{array}
\end{eqnarray*}
\end{corollary}

\subsubsection{\large Moving from one End to the Other is the most Costly}
The following statement shows that \texttt{GeneralMove} requires the maximal number of moves when the source
and destination pegs are at the extreme ends of the set $A$.
\begin {proposition}\label{l:farh}
For integers $h \ge 3,n \ge 1,s,d$, such that $1 \! \le \! s \! < \! d \! \le \! h$ and $d - s + 1 < h$,
$$G_{s \rightarrow d}(h,n) < F(h,n).$$
\end {proposition}
\begin {proof}
Denote $h''=h-h'+1$. The proof is by induction on $n$, for all values of $h \ge 3$.
\\ For $n=1$, we have $\widetilde{n}_j = 0$ for each $1 \le j \le h-2$.
Hence by Lemma~\ref{general06}, $$G_{s \rightarrow d}(h,1)
~=~ F(h'',1) ~=~ h''-1 ~=~ d-s ~<~ h-1 ~=~ F(h,1).$$
We assume that the statement holds for less than $n$ disks and all $h \ge 3$, and prove it for $n$ disks and all $h \ge 3$. Observe that
$$1 ~<~ s+h-d ~=~ h' ~\le~ h-1.$$
By Lemma~\ref{partition03}, for each $1 \le j \le h'-1$, we have
$\widetilde{n}_j < n$. Thus, by Lemma~\ref{general06} and the
induction hypothesis,
\begin{eqnarray*}
G_{s \rightarrow d}(h,n) &\le&
\sum_{j=1}^{h'-1}  2F(h-j+1,\widetilde{n}_j) + F(h'',n(h'))
 \\  &<&  \sum_{j=1}^{h'-1}  2F(h-j+1,\widetilde{n}_j) + F(h'',n(h'))
 + (h'-1).
\end{eqnarray*}
Since $h' \le h-1$ and $\widetilde{n}_{h-1} = 1$:
$$\sum_{j=h'}^{h-1} G_{1 \rightarrow 2}(h-j+1,\widetilde{n}_j) ~\ge~
G_{1 \rightarrow 2}(2,\widetilde{n}_{h-1}) ~=~ 1.$$
Thus by Corollary~\ref{c:farh},
\begin {eqnarray*}
\begin{array}{lll}
F(h,n) & \ge & \displaystyle{\sum_{j=1}^{h'-1} 3F(h-j+1,\widetilde{n}_j)} + F(h'',n(h')) + (h'-1)\sum_{j=h'}^{h-1} G_{1 \rightarrow 2}(h-j+1,\widetilde{n}_j) \cr
       & \ge & \displaystyle{\sum_{j=1}^{h'-1} 2F(h-j+1,\widetilde{n}_j) + F(h'',n(h')) + (h'-1)} \cr
       &  > & G_{s \rightarrow d}(h,n). $\QED$
\end{array}
\end {eqnarray*}
\end {proof}

\subsection{Proof of Theorem~\ref{tp1}\label{S:Up}}

\subsubsection{\large Auxiliary statements}
\begin{lemma}\label{ratio09}
For any integers $h \ge 4$ and $n \ge 1$,
$$G_{1 \rightarrow 2}(h,n) \le \frac{2}{3}F(h,n)-1.$$
\end{lemma}
\begin {proof}
By Lemma~\ref{general06},
\begin{eqnarray*}
G_{1 \rightarrow 2}(h,n)  &=&  \sum_{j=1}^{h-2}
\left[F(h-j+1,\widetilde{n}_j) + G_{2 \rightarrow
h-j+1}(h-j+1,\widetilde{n}_j)\right] ~+~ F(2,n(h-1)).
\end{eqnarray*}
Note that $F(2,n(h-1)) =1$. Thus by Proposition \ref{l:farh},
\begin{eqnarray*}
G_{1 \rightarrow 2}(h,n)  &\le&  \sum_{j=1}^{h-2} 2F(h-j+1,\widetilde{n}_j)+1.
\end{eqnarray*}
Note that $h-1 \ge 3$. By Corollary \ref{c:farh} in the particular
case $k=h-1$,
\begin {eqnarray*}
F(h,n) &\ge& \sum_{j=1}^{h-2} 3F(h-j+1,\widetilde{n}_j)
 + F(2,n(h-1)) + (h-2)G_{1 \rightarrow 2}(2,\widetilde{n}_{h-1}) \cr
       &\ge& \sum_{j=1}^{h-2} 3F(h-j+1,\widetilde{n}_j) + 1 + h-2 \cr
       &\ge& \sum_{j=1}^{h-2} 3F(h-j+1,\widetilde{n}_j) + 3.
\end{eqnarray*}
Altogether,
\begin{eqnarray*}
G_{1 \rightarrow 2}(h,n) ~\le~ \sum_{j=1}^{h-2} 2F(h-j+1,\widetilde{n}_j)
 + 1 ~\le~ \frac{2}{3}F(h,n) - 1.
 \end{eqnarray*}
\QED \end {proof}

\begin{lemma}~\label{tshieit10}
For any integers $h \ge 4$ and $n \ge 2$,
$$ \sum_{j=1}^{h-1}G_{1 \rightarrow 2}(h-j+1,\widetilde{n}_j) \le \frac{2}{9}F(h,n).$$
\end{lemma}
\begin{proof}
By Lemma~\ref{farthest07},
$$F(h,n) = \sum_{j=1}^{h-2} \left[ 3F(h-j+1,\widetilde{n}_j)
+ (j-1)G_{1 \rightarrow 2}(h-j+1,\widetilde{n}_j) \right] +
h-1.$$
Note that $\frac{2}{9}(h-1) \ge \frac{2}{3}$. Hence,
\begin{eqnarray} \label{e:fn}
\frac{2}{9}F(h,n) &\ge&
 \frac{2}{3} \left(\sum_{j=1}^{h-2} F(h-j+1,\widetilde{n}_j)
 + 1\right) = \frac{2}{3}\sum_{j=1}^{h-1}
 F(h-j+1,\widetilde{n}_j).
 \end{eqnarray}
We claim that
 \begin{equation} \label{neweq}
 G_{1 \rightarrow 2}(h-j+1,\widetilde{n}_j) ~\le~ \frac{2}{3}
F(h-j+1,\widetilde{n}_j), \qquad 1 \le j \le h-3.
\end {equation}
To this end, note that by Lemma~\ref{ratio09} this clearly holds if
$\widetilde{n}_j \ge 1$. Otherwise $\widetilde{n}_j = 0$, so
$$G_{1 \rightarrow 2}(h-j+1,\widetilde{n}_j) = 0 = \frac{2}{3}F(h-j+1,\widetilde{n}_j).$$

Since $n \ge 2$ and $\widetilde{n}_{h-1} = 1$, there exists a number $j$ with
$1 \le j \le h-2$ such
 that $\widetilde{n}_j \ge 1$.
\\The analysis splits into two cases.
\\\emph{Case 1: $\widetilde{n}_{h-2} \ge 1$.}

By (\ref{neweq}),
$$ \sum_{j=1}^{h-3} G_{1 \rightarrow 2}(h-j+1,\widetilde{n}_j) ~\le~ \frac{2}{3}
\sum_{j=1}^{h-3} F(h-j+1,\widetilde{n}_j).$$
Note that
$F(3,\widetilde{n}_{h-2}) \ge 2$. Thus, by Lemma~\ref{classical05},
$$G_{1 \rightarrow 2}(3,\widetilde{n}_{h-2}) ~=~ \frac{1}{2}F(3,\widetilde{n}_{h-2}) ~\le~ \frac{2}{3}F(3,\widetilde{n}_{h-2}) - \frac{1}{3}.$$
Altogether, we have
\begin{eqnarray} \label{e:case1}
\begin{array}{lll}
\displaystyle{\sum_{j=1}^{h-1} G_{1 \rightarrow 2}(h-j+1,\widetilde{n}_j)}
 & = & \displaystyle{\sum_{j=1}^{h-3} G_{1 \rightarrow 2}(h-j+1,\widetilde{n}_j)
 + G_{1 \rightarrow 2}(3,\widetilde{n}_{h-2}) + 1} \cr
 &\le& \displaystyle{\frac{2}{3}\sum_{j=1}^{h-3} F(h-j+1,\widetilde{n}_j) + \frac{2}{3}F(3,\widetilde{n}_{h-2}) + \frac{2}{3}} \cr
 & = & \displaystyle{\frac{2}{3}\sum_{j=1}^{h-1} F(h-j+1,\widetilde{n}_j)}.
\end{array}
\end{eqnarray}
By (\ref{e:fn}), the right-hand side of (\ref{e:case1}) is no
greater than $\frac{2}{9}F(h,n)$, as required.
\\\emph{Case 2: $\widetilde{n}_i \ge 1$ for some $i$ where $1 \le i \le h-3$.}

By (\ref{neweq}) and Lemma~\ref{classical05},
$$\sum_{j \in [h-2] \setminus i} G_{1 \rightarrow 2}(h-j+1,\widetilde{n}_j) ~\le~ \frac{2}{3}\sum_{j \in [h-2]
\setminus i} F(h-j+1,\widetilde{n}_j).$$

By Lemma~\ref{ratio09},
$$ G_{1 \rightarrow 2}(h-i+1,\widetilde{n}_i) ~\le~ \frac{2}{3}
F(h-i+1,\widetilde{n}_i)-1.$$
Altogether, we have
\begin{eqnarray} \label{e:case2}
\begin{array}{lll}
\displaystyle{\sum_{j=1}^{h-1} G_{1\rightarrow 2}(h-j+1,\widetilde{n}_j)}
 & = & \displaystyle{\sum_{j \in [h-2] \setminus \{i\}} G_{1 \rightarrow 2}(h-j+1,\widetilde{n}_j)}
 + G_{1 \rightarrow 2}(h-i+1,\widetilde{n}_{i}) + 1 \cr
 & \le & \displaystyle{\frac{2}{3}\sum_{j=1}^{h-2} F(h-j+1,\widetilde{n}_j)}.
\end{array}
\end{eqnarray}
By (\ref{e:fn}), the right-hand side of (\ref{e:case2}) is strictly
less than $\frac{2}{9}F(h,n)$, and we are done. \QED
\end{proof}
\begin{lemma}\label{nlargeh11}
For any integers $h \ge 5$ and $n \ge h$,
$$F(h,n) \le 3F(h,\widetilde{n}_1) + \frac{11}{9}F(h-1,n-\widetilde{n}_1).$$
\end{lemma}
\begin{proof}
By Corollary \ref{c:farh} in the particular case $k=2$,
\begin{eqnarray*}
F(h,n) &=& 3F(h,\widetilde{n}_1) ~+~
 F(h-1,n-\widetilde{n}_1) +  \sum_{j=2}^{h-1} G_{1 \rightarrow 2}(h-j+1,\widetilde{n}_j).
\end{eqnarray*}
We have $n \ge h \ge 5$. Thus, $h-1 \ge 4$, and by Lemma~\ref{remainder02} we have $n-\widetilde{n}_1 \ge 2$. Applying
Lemma~\ref{tshieit10} with $h-1$ and $n-\widetilde{n}_1$ instead of $h$ and $n$, respectively, we get
\begin{equation} \label{e:desired} \sum_{j=1}^{h-2} G_{1 \rightarrow
2}(h-j,\widetilde{n}_j(h-1,n-\widetilde{n}_1)) ~\le~ \frac{2}{9}F(h-1,n-\widetilde{n}_1).
\end{equation}
By Lemma~\ref{blocksize04} in the particular case
$i=2$, for each $1 \le j \le h-2$,
 \begin {eqnarray*}
 \widetilde{n}_{j+1}(h,n) ~=~ \widetilde{n}_{j}(h-1,n-\widetilde{n}_1).
 \end {eqnarray*}
Consequently, $$\sum_{j=2}^{h-1} G_{1 \rightarrow
2}(h-j+1,\widetilde{n}_j) ~=~ \sum_{j=1}^{h-2} G_{1 \rightarrow
2}(h-j,\widetilde{n}_j(h-1,n-\widetilde{n}_1)) ~\le~ \frac{2}{9}F(h-1,n-\widetilde{n}_1),$$  which provides the required
result.
\end{proof}

\begin{lemma}\label{nlessh12}
For any integers $h \ge 5$ and $1 \le n < h$,
$$F(h,n) \le \widetilde{U}(h,n) = \widetilde{C}_{h}\cdot n^{\alpha_{h}}\cdot 3^{\theta_{h}\cdot n^{\frac{1}{h-2}}} ~<~ U(h,n) ,$$
where $\theta_{h}$ and $\alpha_{h}$ are as in Theorem~\ref{tp1}, and $\widetilde{C}_{h} = \frac{h-2}{\theta_{h}}$.
\end{lemma}
\begin {proof}
First note that $\widetilde{U}(h,n) \!<\! U(h,n)$, so it remains to
prove that $F(h,n) \le \widetilde{U}(h,n)$.
Since $h \ge 5$ and $n < h$, Lemmas~\ref{remainder02} and \ref{partition03}
imply that $\widetilde{n}_{i} =
1$ for each $i \in [n-1] \cup \{h-1\}$. Thus by Lemma~\ref{farthest07},
\begin{eqnarray*}
F(h,n) &=&
 \sum_{j=1}^{h-2} \left[3F(h-j+1,\widetilde{n}_j)
+ (j-1)G_{1 \rightarrow 2}(h-j+1,\widetilde{n}_j) \right] + h-1  \cr
 &=& \sum_{j=1}^{n-1} \left[3F(h-j+1,1)
+ (j-1)G_{1 \rightarrow 2}(h-j+1,1) \right] + h-1 \cr
 &=& \sum_{j=1}^{n-1} 3(h-j) + (j-1) + h-1 ~=~ n(3h-n)-2h.
\end{eqnarray*}
It is easy to verify that for $h \ge 5$ and $n < h$,
$$n(3h-n) -2h \le 3n(h-2).$$
Note that $\theta_{h}n^{\frac{1}{h-2}} \ge 1$ for $h \ge 3$ and $n \ge 1$. For $x \ge 1$, we have $x \le 3^{x-1}$. Therefore,
$$3\theta_{h}n^{\frac{1}{h-2}} \le 3^{\theta_{h}n^{\frac{1}{h-2}}},$$
and consequently,
$$3n(h-2) \le \frac{h-2}{\theta_{h}}n^{\alpha_{h}} \cdot
3^{\theta_{h}n^{\frac{1}{h-2}}}.$$
Altogether,
$$ F(h,n) ~\le~
3n(h-2) \le \frac{h-2}{\theta_{h}}n^{\alpha_{h}}\cdot 3^{\theta_{h}\cdot n^{\frac{1}{h-2}}} ~=~ \widetilde{U}(h,n).$$ \QED
\end {proof}

\subsubsection{\large Conclusion of the Proof}
The proof is by double induction on $h \ge 3$ and $n$.

For $h=3$, the algorithm works exactly as does the algorithm of
\cite{ScGrSm44} for the {\rm{3-in-a-row}} graph. Therefore, the
number $F(3,n)$ of moves required by this algorithm for $n$ disks
 is $3^n-1$. The
substitution $h=3$ in the upper bound $U(h,n)$ suggested by the
proposition yields:
$$U(3,n) ~=~ C_h n^{\alpha_h} \cdot 3^{\theta_h n^{1/\frac{1}{h-2}}} ~=~ 1 \cdot n^0 \cdot 3^{1 \cdot n} ~=~ 3^n ~>~
F(3,n).$$
For $h=4$, the algorithm works exactly as does the
algorithm $\texttt{FourMove}$ of Section \ref{h4} for moving $n$ disks
between the two farthest pegs in \texttt{Path$_4$}. Therefore, as
shown in Section \ref{h4}, $F(4,n)$ is bounded above by $1.6 \sqrt{n} \cdot
3^{\sqrt{2n}}$. The substitution $h=4$ in the upper bound $U(h,n)$
suggested by the proposition yields:
$$U(4,n) ~=~ C_h n^{\alpha_h} \cdot 3^{\theta_h n^{\frac{1}{h-2}}} = c\sqrt{2n} \cdot 3^{\sqrt{2n}} > 1.6 \sqrt{n} \cdot
3^{\sqrt{2n}} \ge F(4,n).$$
For $h \ge 5$ and $n < h$, Lemma~\ref{nlessh12} implies that $F(h,n)
< U(h,n)$.

We assume that for arbitrary fixed $h \ge 5$ and $n \ge h$,
$F(h',n') < U(h',n')$ holds for all $(h{'},n{'})$ with either $h{'}
< h$ or both $h{'} = h$ and $n{'} < n$, and prove it for $(h,n)$.

Let $m = m_1 = \texttt{Remainder}(h,n)$. Then $\widetilde{n}_1 = n-m$. By Lemma~\ref{nlargeh11},
$$F(h,n) \le 3F(h,n-m) + \frac{11}{9}F(h-1,m).$$
The analysis splits into two cases. \\\emph
{Case 1: $n \le \frac{(h-2)^{h-2}}{(h-2)!}$.}

In this case, we have
\begin{equation} \label{eq:ss} n ~\le~
\frac{((h-2)!)^{\alpha_h}}{(h-3)!} n^{\alpha_h}, \end
{equation} and so,
$$m ~=~{\rm\texttt{Remainder}}(h,n) ~=~
\min\left\{n,\left\lceil\frac{((h-2)!)^{\alpha_h}}{(h-3)!}
n^{\alpha_h}\right\rceil\right\} ~=~ n.$$ It follows that $F(h,n)
\le \frac{11}{9}F(h-1,n).$ By the induction hypothesis,
\begin{eqnarray} \label{eq111}
 F(h,n) ~<~ \frac{11}{9}U(h-1,n) ~=~ \frac{11}{9}C_{h-1} \cdot n^{\alpha_{h-1}} \cdot 3^{\theta_{h-1}\cdot
 n^{\frac{1}{h-3}}}.
\end{eqnarray}
By (\ref{eq:ss}), we have
\begin{eqnarray}~\label{eq141}
\theta_{h-1}\cdot n^{\frac{1}{h-3}} \le
\theta_{h-1}\!\left(\frac{((h-2)!)^{\alpha_h}}{(h-3)!}
n^{\alpha_h}\right)^\frac{1}{h-3} =
\theta_{h-1}\!\left(\frac{\theta_{h}^{h-3}}{\theta_{h-1}^{h-3}}
n^{\frac{h-3}{h-2}} \right)^{\frac{1}{h-3}} = \theta_{h}\cdot
n^{\frac{1}{h-2}}.
\end{eqnarray}
Substituting (\ref{eq141}) in (\ref{eq111}), we obtain
\begin{eqnarray}~\label{eq151}
F(h,n) < \frac{11}{9}C_{h-1} \cdot n^{\alpha_{h-1}} \cdot
3^{\theta_{h} \cdot n^{\frac{1}{h-2}}}.
\end{eqnarray}
It is easy to verify that $\frac{11}{9}C_{h-1} < C_h$ and
$\alpha_{h-1} < \alpha_h$. Thus we find that:
$$ F(h,n) ~<~ C_{h} \cdot n^{\alpha_{h}} \cdot 3^{\theta_{h} \cdot n^{\frac{1}{h-2}}}.
$$
\\\emph{Case 2: $n > \frac{(h-2)^{h-2}}{(h-2)!}$.}

In this case, we have
\begin
{eqnarray*} n ~>~ \frac{((h-2)!)^{\alpha_h}}{(h-3)!}
n^{\alpha_h}, \end {eqnarray*} and so,
$$m \!=\!
{\rm\texttt{Remainder}}(h,n) \!=\!
\min\left\{n,\left\lceil\frac{((h-2)!)^{\alpha_h}}{(h-3)!}
n^{\alpha_h}\right\rceil \right\} \!=\!
\left\lceil\frac{((h-2)!)^{\alpha_h}}{(h-3)!}
n^{\alpha_h}\right\rceil \!\le\! n.$$ By the induction hypothesis,
\begin{eqnarray}~\label{eq11}
F(h,n) ~<~  3C_{h}(n-m)^{\alpha_{h}} 3^{\theta_{h}\cdot
(n-m)^{\frac{1}{h-2}}}
              +\frac{11}{9}C_{h-1} \cdot m^{\alpha_{h-1}} \cdot 3^{\theta_{h-1} \cdot
              m^{\frac{1}{h-3}}}.
\end{eqnarray}
Observe that
\begin{eqnarray}~\label{eq12}
 (n-m)^{\alpha_{h}}  ~=~ n^{\alpha_{h}}\left(1-\frac{m}{n}\right)^{\alpha_{h}}
                     ~<~ n^{\alpha_{h}}\left(1-\frac{\alpha_{h} m}{n}\right).
\end{eqnarray}
Similarly, we have
\begin{eqnarray}~\label{eq13}
 \theta_{h}(n\!-\!m)^{\frac{1}{h-2}} = \theta_{h}
n^{\frac{1}{h-2}}\left(1\!-\!\frac{m}{n}\right)^\frac{1}{h-2} \!<\!
\theta_{h}
n^{\frac{1}{h-2}}\left(1\!-\!\frac{\theta_{h}^{h-3}\cdot n^{\frac{h-3}{h-2}}} {(h-2) n \theta_{h-1}^{h-3}}\right)
 \!\!=\! \theta_{h} n^{\frac{1}{h-2}} \!-\! 1
\end{eqnarray}
and
\begin{eqnarray*}
 \theta_{h-1} \cdot m^{\frac{1}{h-3}} &\le&
\theta_{h-1}\left(\frac{\theta_{h}^{h-3}}{\theta_{h-1}^{h-3}}\cdot
n^{\frac{h-3}{h-2}} + 1\right)^{\frac{1}{h-3}}
 ~=~ \theta_{h} n^{\frac{1}{h-2}}\left(1 + \frac{\theta_{h-1}^{h-3}}{\theta_{h}^{h-3}}
 n^{-\frac{h-3}{h-2}}\right)^{\frac{1}{h-3}} \cr
  &<& \theta_{h} n^{\frac{1}{h-2}} + \frac{(h-3)! \cdot \theta_{h}^{2}}{(h\!-\!3)(h\!-\!2)!} n^{-\frac{h-4}{h-2}}
 ~=~
 \theta_{h} n^{\frac{1}{h-2}}+\frac{\theta_{h}^{2}}{(h\!-\!3)(h\!-\!2)} n^{-\frac{h-4}{h-2}}.
\end{eqnarray*}
It is easy to verify that $\vartheta(h,n) =
\frac{\theta_{h}^{2}}{(h-3)(h-2)} n^{-\frac{h-4}{h-2}}$ is
monotone decreasing with $h$ and $n$ in the range $n \ge h \ge 5$.
Hence for $n \ge h \ge 5$, we have
$$\vartheta(h,n) ~\le~ \vartheta(5,5) ~=~ \left(\frac{1}{30}\right)^{1/3}.$$
Put $c^* = \left(\frac{1}{30}\right)^{1/3}$.
 Now
\begin{eqnarray}~\label{eq14}
 \theta_{h-1}\cdot m^{\frac{1}{h-3}}  &<&
 \theta_{h} \cdot n^{\frac{1}{h-2}} + c^*.
\end{eqnarray}

 Substituting (\ref{eq12}), (\ref{eq13}) and (\ref{eq14}) in
(\ref{eq11}), we obtain
\begin{eqnarray}~\label{eq15}
\begin{array}{lll}
F(h,n) & < & 3C_{h} n^{\alpha_{h}}\left(1-\frac{\alpha_{h}\cdot
m}{n}\right)3^{\theta_{h}\cdot n^{\frac{1}{h-2}}-1}
    + \frac{11}{9} C_{h-1} \cdot m^{\alpha_{h-1}} \cdot 3^{\theta_{h} \cdot n^{\frac{1}{h-2}} + c^{*}} \cr
       & = & C_{h} n^{\alpha_{h}}\left(1-\frac{\alpha_{h}\cdot m}{n}\right)3^{\theta_{h}\cdot n^{\frac{1}{h-2}}}
 + \frac{11}{9} 3^{c^{*}}C_{h-1} \cdot m^{\alpha_{h-1}} \cdot 3^{\theta_{h} \cdot n^{\frac{1}{h-2}}} \cr
       & = & \left(C_{h} n^{\alpha_{h}} - \frac{C_{h}\cdot n^{\alpha_{h}}\cdot \alpha_{h}\cdot m}{n}
   + \delta \cdot C_{h-1}\cdot m^{\alpha_{h-1}}\right) 3^{\theta_{h}\cdot n^{\frac{1}{h-2}}}.
\end{array}
\end{eqnarray}

The second and third terms on the right-hand side of (\ref{eq15}) may be omitted
since:
$$\begin{array}{lll}
\delta C_{h-1}\cdot m^{\alpha_{h-1}} -
\frac{C_{h}\cdot n^{\alpha_{h}}\cdot \alpha_{h}\cdot m}{n} & = & m^{\alpha_{h-1}}\left(\delta\frac{(h-3)\delta^{h-4}}{\theta_{h-1}} -
 \frac{(h-2)\delta^{h-3}}{n \cdot \theta_{h}}\cdot n^{\frac{h-3}{h-2}} \cdot \frac{h-3}{h-2}
\cdot m^{\frac{1}{h-3}}\right) \cr
 & = & m^{\alpha_{h-1}}(h-3)\delta^{h-3}\left(\frac{1}{\theta_{h-1}} -
 \frac{n^{\frac{-1}{h-2}}\cdot m^{\frac{1}{h-3}}}{\theta_{h}}\right) \cr
 & \le & m^{\alpha_{h-1}}(h\!-\!3)\delta^{h-3}\left(\frac{1}{\theta_{h-1}} \!-\!
 \frac{n^{\frac{-1}{h-2}}}{\theta_{h}}
 \left(\frac{\theta_{h}^{h-3}}{\theta_{h-1}^{h-3}}\cdot n^{\frac{h-3}{h-2}}
 \right)^{\frac{1}{h-3}}\right) \cr
 & = & 0.
\end{array}$$
Thus we conclude that:
$$\begin{array}{l}
 F(h,n) ~<~ C_{h} n^{\alpha_{h}} \cdot 3^{\theta_{h}\cdot n^{\frac{1}{h-2}}}.
\end{array}$$\QED

\end {document}